\documentclass[envcountsame]{llncs}
\usepackage{a4wide}
\usepackage[ruled, vlined]{algorithm2e}
\usepackage{amsfonts,amssymb,amstext,amsmath}
\usepackage{booktabs}
\usepackage{subcaption}
\usepackage{enumerate}
\usepackage{floatrow}
\usepackage{fancyvrb}
\usepackage{graphicx}
\usepackage{thmtools,thm-restate}
\usepackage{tikz}
\usetikzlibrary{arrows, petri, topaths, automata, shapes,decorations.shapes}
\usepackage{tkz-berge}
\usepackage[textsize=tiny]{todonotes}
\usepackage{xspace}

\newcommand{\onemaxm}{$\mathsf{OneMax}$\xspace}

\newcommand{\ea}{(1+1)~EA\xspace}

\newcommand{\bigo}[1]{\mathbf{\mathcal{O}} \left ( #1 \right)}

\newcommand{\onemax}[1]{\left \lvert #1\right \rvert_1}

\newcommand{\pr }[1]{\mathsf{Pr}\left ( #1 \right )}

\newcommand{\expect}[1]{\mathbb{E}  \left [ #1 \right ]}

\newcommand{\opt}{\textsc{opt}\xspace}

\newcommand{\absl}[1]{\left \lvert #1 \right \rvert}

\newcommand{\gsemo}{\mbox{GSEMO}\xspace}

\newcommand{\ind}{\mathcal{I}}

\newcommand{\jump}{\mathsf{Jump}_{m, n}\xspace}
\newcommand{\plow}{\mathsf{fmut}_{\beta}\xspace}
\newcommand{\mymut}{\mathsf{pmut}_{\beta}\xspace}
\newcommand{\mycmut}{$\mbox{\textsf{cMut}}(p)$\xspace}
\newcommand{\mut}[1]{\mymut}
\newcommand{\umut}{\ensuremath{\mathsf{unif}_p}}
\newcommand{\stdmut}{\ensuremath{\mathsf{unif}_{1}}}
\newcommand{\clmut}{\mathsf{pmut}_{\beta}\xspace}
\newcommand{\pmut}[1]{\ensuremath{\mathsf{pmut}_{#1}}}
\newcommand{\fmut}[1]{\ensuremath{\mathsf{fmut}_{#1}}}

\newcommand{\maxdicut}{\mbox{$\mathsf{Max}$-$\mathsf{Di}$-$\mathsf{Cut}$}\xspace}
\newcommand{\hamming}[1]{\mathcal{H} \left ( #1 \right )}
\newcommand{\inverse}[1]{( #1 )^{-1}}
\newcommand{\optmut}{\textsc{opt}\xspace}
\newcommand{\lm}{S}

\newcommand{\mutualinfo}[1]{\mathsf{MI} \left (#1 \right )}
\newcommand{\notwo}[1]{NO$_2$\xspace}

\DeclareMathOperator*{\argmax}{arg\,max}

\renewenvironment{proof}{\paragraph{Proof}}{\hfill$\square$\vspace{10pt}\\}

\begin{document}

\title{Evolutionary Algorithms and Submodular Functions: Benefits of Heavy-Tailed Mutations}

\titlerunning{Evolutionary Algorithms and Submodular Functions: Benefits of Heavy-Tailed Mutations} 

\author{Tobias Friedrich \inst{1,2} \and Andreas G\"{obel} \inst{1} \and Francesco Quinzan \inst{1} \and Markus Wagner \inst{2}}

\authorrunning{Chauhan et al.}
\tocauthor{Ivar Ekeland, Roger Temam, Jeffrey Dean, David Grove,
Craig Chambers, Kim B. Bruce, and Elisa Bertino}
\institute{Hasso Plattner Institute, Potsdam, Germany
\and University of Adelaide, Adelaide, Australia
}

\maketitle         

\begin{abstract}
A core feature of evolutionary algorithms is their mutation operator. Recently, much attention has been devoted to the study of mutation operators with dynamic and non-uniform mutation rates. Following up on this line of work, we propose a new mutation operator and analyze its performance on the (1+1) Evolutionary Algorithm (EA).

Our analyses show that this mutation operator competes with pre-existing ones, when used by the \ea on classes of problems for which results on the other mutation operators are available. We show that the \ea using our mutation operator finds a $(1/3)$-approximation ratio on any non-negative submodular function in polynomial time. We also consider the problem of maximizing a symmetric submodular function under a single matroid constraint and show that the \ea using our operator finds a  $(1/3)$-approximation within polynomial time. This performance matches that of combinatorial local search algorithms specifically designed to solve these problems and outperforms them with constant probability.

Finally, we evaluate the performance of the \ea using our operator experimentally by considering two applications: (a) the maximum directed cut problem on real-world graphs of different origins, and with up to 6.6 million vertices and 56 million edges and (b) the symmetric mutual information problem using a four month period air pollution data set. In comparison with uniform mutation and a recently proposed dynamic scheme our operator comes out on top on these instances.
\end{abstract}
%
%
\section{Introduction}
A key procedure of the \ea that affects its performance is the \emph{mutation operator}, i.e., the operator that determines at each step how the potential new solution is generated. In the past several years there has been a huge effort, both from a theoretical and an experimental point of view, towards understanding how this procedure influences the performance of the \ea and which is the optimal way of choosing this parameter (e.g., see~\cite{Eiben1999parametercontrol,Eiben2003book}).

The most common mutation operator on $n$-bit strings is the static \emph{uniform mutation} operator. This operator, \umut{}, flips each bit of the current solution independently with probability~$p(n)$. This probability, $p(n)$, is called static \emph{mutation rate} and remains the same throughout the run of the algorithm. The most common choice for $p(n)$ is $1/n$; thus, mutated solutions differ in expectation in one bit from their predecessors. Witt~\cite{Witt:2005:WAA:2140048.2140054} shows that this choice of $p(n)$ is optimal for all pseudo-Boolean linear functions. Doerr~et~al.~\cite{doerr2013mutationratematters} further observe that changing $p(n)$ by a constant factor can lead to large variations of the overall run-time of the \ea. They also show the existence of functions for which this choice of $p(n)$ is not optimal.

Static mutation rates are not the only ones studied in literature. Jansen et al. \cite{DBLP:journals/dam/JansenW05} propose a mutation rate which at time step $t$ flips each bit independently with probability $2^{(t - 1) \mod (\lceil \log_2 n\rceil - 1)}/n$. Doerr et al.~\cite{DBLP:conf/gecco/DoerrLMN17} observe that this mutation rate is equivalent to a mutation rate of the form $\alpha /n$, where $\alpha$ is chosen uniformly at random (u.a.r.) from the set $\{2^{(t - 1) \mod (\lceil \log_2 n\rceil - 1)}\mid t\in\{1,\dots,\lceil \log_2 n\rceil\} \}$. 
Doerr et al.~\cite{doerr2018onthefly,doerr2018sensitivity} have proposed a simple on-the-fly mechanism that can approximate optimal mutation rates well for two unimodal functions.

Doerr~et~al.~\cite{DBLP:conf/gecco/DoerrLMN17} notice that the choice of $p(n)=1/n$ is a result of over-tailoring the mutation rates to commonly studied simple unimodal problems. They propose a non-static mutation operator $\plow$, which chooses a mutation rate $\alpha \leq 1/2$ from a power-law distribution at every step of the algorithm. Their analysis shows that for a family of ``jump'' functions introduced below, the run-time of the \ea yields a polynomial speed-up over the optimal time when using $\plow$.

Friedrich~et~al.~\cite{Friedrich2018heavytailedGECCO} propose a new mutation operator. Their operator~\mycmut chooses at each step with constant probability $p$ to flip 1-bit of the solution chosen uniformly at random. With the remaining probability $1-p$, the operator chooses $k\in\{2,\dots,n\}$ uniformly at random and flips~$k$ bits of the solution chosen uniformly at random. This operator performs well in optimizing pseudo-Boolean functions, as well as combinatorial problems such as the minimum vertex cover and the maximum cut. Experiments suggest that this operator outperforms the mutation operator of Doerr et al.~\cite{DBLP:conf/gecco/DoerrLMN17} when run on functions that exhibit large deceptive basins of attraction, i.e., local optima whose Hamming distance from the global optimum is in $\Theta(n)$.

As evolutionary algorithms are used extensively in real world applications, it is important to extend the theoretical analysis of their performance to the more general classes of functions. To improve the performance of \ea in more complex landscapes and inspired by the recent results of Doerr~et~al.~\cite{DBLP:conf/gecco/DoerrLMN17} and Friedrich~et~al.~\cite{Friedrich2018heavytailedGECCO} we propose a new mutation operator $\mymut$. Our operator mutates $n$-bit string solutions as follows. At each step, $\mymut$ chooses $k\in\{1,\dots,n\}$ from a power-law distribution. Then $k$ bits of the current solution are chosen uniformly at random and then flipped. During a run of the \ea using $\mymut$, the majority of mutations consist of flipping a small number of bits, but occasionally a large number, of up to $n$ bit flips can be performed. In comparison to the mutations of $\plow$, the mutations of $\mymut$ have a considerably higher likelihood of performing larger than $(n/2)$-bit jumps.  

\subsubsection*{Run-Time Comparison on Artificial Landscapes} 
Our analysis of the \ea using $\mymut$ starts by considering artificial landscapes. More specifically, in Section~\ref{theoretical:clmut} we show that the \ea using $\mymut$ manages to find the optimum of any function within exponential time. When run on the OneMax function, the \ea with $\mymut$ finds the optimum solution in expected polynomial time. 

In Section~\ref{sec:jump} we consider the problem of maximizing the $n$-dimensional jump function $\jump(x)$, first introduced by Droste et al. \cite{Droste:2002:AEA:568273.568277}. We show that for any value of the parameters $m,n$ with $m$ constant or $n - m$, the expected run time of the \ea using $\mymut$ remains polynomial. This is not the case for the \ea using \umut{}, for which Droste et al. \cite{Droste:2002:AEA:568273.568277} showed a run time of $\Theta (n^{m} + n \log n)$ in expectation. Doerr et al.~\cite{DBLP:conf/gecco/DoerrLMN17} are able to derive polynomial bounds for the expected run-time of the \ea using their mutation operator $\plow$, but in their results limit the jump parameter to $m \leq n/2$. 

\subsubsection*{Optimization of Submodular Functions}
Our main focus in this article is to study the performance of the \ea when optimizing submodular functions. Submodularity is a property that captures the notion of diminishing returns. Thus submodular functions find applicability in a large variety of problems. Examples include: maximum facility location problems~\cite{DBLP:journals/dam/AgeevS99}, maximum cut and maximum directed cut~\cite{Goemans:1995:IAA:227683.227684}, restricted $\mathsf{SAT}$ instances~\cite{DBLP:journals/jacm/Hastad01}. Submodular functions under a single matroid constraint arise in artificial intelligence and are connected to probabilistic fault diagnosis problems \cite{DBLP:conf/aaai/KrauseG07,Lee:2009:NSM:1536414.1536459}.

Submodular functions exhibit additional properties in some cases, such as \emph{symmetry} and \emph{monotonicity}. These properties can be exploited to derive run time bounds for local randomized search heuristics such as the \ea. In particular, Friedrich and Neumann \cite{DBLP:journals/ec/FriedrichN15} give run time bounds for the \ea and \gsemo on this problem, assuming either monotonicity or symmetry.

We show (Section~\ref{sec:submodular-bound}) that the \ea with $\mymut$ on any non-negative, submodular function gives a $1/3$-approximation within polynomial time. This result matches the performance of the local search heuristic of Feige et al.~\cite{DBLP:journals/siamcomp/FeigeMV11} designed to target non-negative, submodular functions in particular. An example of a natural non-negative submodular function that is neither symmetric nor monotone is the utility function of a player in a combinatorial auction (see e.g.~\cite{LLN06}). We further show  (Section~\ref{sec:large_populations_unconstrained}) that the \ea{} outperforms the local search of Feige et al.~\cite{DBLP:journals/siamcomp/FeigeMV11} at least with constant probability (w.c.p.).

Additionally we evaluate the performance of the \ea on the maximum directed cut problem using $\mymut$ experimentally, on real-world graphs of different origins, and with up to 6.6 million vertices and 56 million edges. Our experiments show that $\mymut$ outperforms $\umut$ and the uniform mutation operator on these instances. This analysis appears in Section~\ref{sec:max_di_cut}

In section~\ref{sec:submodular_const} we consider the problem of maximizing a symmetric submodular function under a single matroid constraint. Our analysis shows that the \ea{} using $\mymut$ finds a $1/3$-approximation within polynomial time. Our analysis can be easily extended to show that the same results apply to the \ea when using the uniform mutation operator or $\umut$. 

To establish our results empirically, in Section~\ref{sec:symmetric_mutual_information} we consider the symmetric mutual information problem under a cardinality constraint. We consider an air pollution data set during a four month interval and use the \ea to identify the highly informative random variables of this data set. We observe that $\mymut$ performs better than the uniform mutation operator and $\umut$ for a small time budged and a small cardinality constraint, but for a large cardinality constraint all mutation operators have similar performance.

A comparison of the previously known performance of deterministic local search algorithms on submodular functions and our results on the \ea can be found in Table \ref{fig:table_bounds}.
\begin{table*}[t]
\includegraphics[width=\linewidth]{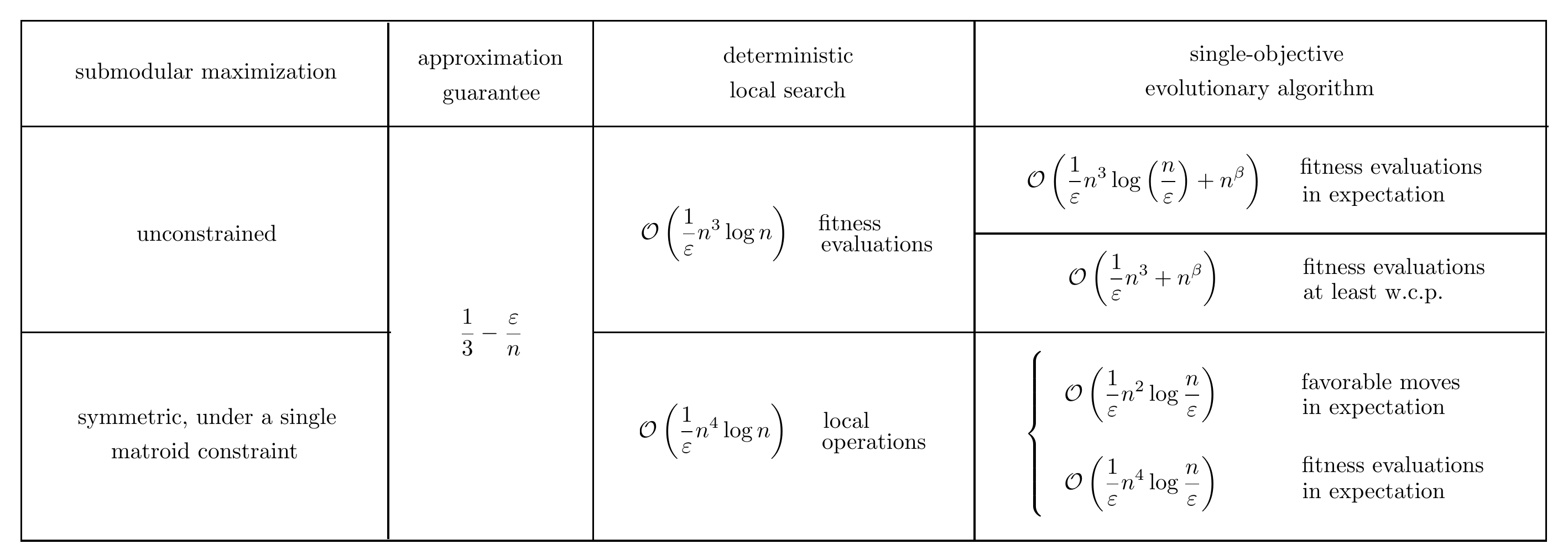}
\caption{Upper bounds on the run time for the \ea with mutation $\clmut$, with parameter $\beta >1$. The expected run time in the unconstrained case is given in Section \ref{sec:submodular-bound}, whereas the improved upper-bound in Section \ref{sec:large_populations_unconstrained}. The expected run time bounds for the \ea in the constrained case are discussed in Section \ref{sec:submodular_const}. Previous run time bounds for deterministic local search algorithms  are discussed in Feige et al.~\cite{DBLP:journals/siamcomp/FeigeMV11} and Lee et al. \cite{Lee:2009:NSM:1536414.1536459}. We remark that \emph{local operations} for the deterministic local search correspond to \emph{favorable moves} in the analysis of the \ea. Hence, they are the same unit of measurement.}
\label{fig:table_bounds}
\end{table*}
%
%
\section{Preliminaries}
\label{algorithms_setting}
\subsection{The \ea and Mutation Rates.}
\begin{algorithm}[t]
	\caption{The \ea}
    \label{alg:ea}
    \textbf{input:} a fitness function $f:2^V \rightarrow \mathbb{R}_{\geq 0}$\;
    \textbf{output:} an (approximate) global maximum\\$\qquad \quad \ \ $of the function $f$\;
    $\qquad $\\
    // sample initial solution\\
 	choose $x\in \{0, 1\}^n$ uniformly at random\;
	$\qquad $\\
 	\While{convergence criterion not met}{
        $\qquad $\\
    	// apply mutation operator\\
   		$y \gets \mathbf{Mutation}(x)$\;
        $\qquad $\\
    	// perform selection\\
		\If{$f(y) \geq f(x)$}{
			$x \gets y$\;
		}
        $\qquad $\\
   	}
    \textbf{return} x\;
\end{algorithm}

We study the run time of the simple $(1+1)$ Evolutionary Algorithm under various configurations. This algorithm requires a bit-string of fixed length~$n$ as input. An offspring is then generated by the \emph{mutation operator}, an operator that resembles asexual reproduction. The fitness of the solution is then computed and the less desirable result is discarded. This algorithm is \emph{elitist} in the sense that the solution quality never decreases throughout the process. Pseudo-code for the \ea is given in Algorithm~\ref{alg:ea}.

In the \ea the offspring generated in each iteration depends on the mutation operator. The standard choice for the $\mathsf{Mutation}(\cdot)$ is to flip each bit of an input string $x = (x_1, \dots, x_n)$ independently with probability $1/n$. In a slightly more general setting, the mutation operator $\umut(\cdot)$ flips each bit of $x$ independently with probability $p/n$, where $p \in [0, n/2]$. We refer to the parameter $p$ as \emph{mutation rate}. 

Uniform mutations can be further generalized, by sampling the mutation rate $p \in [0, n/2]$ at each step according to a given probability distribution. We assume this distribution to be fixed throughout the optimization process. Among this class of mutation rates, is the \emph{power-law} mutation $\plow$ of Doerr et al.~\cite{DBLP:conf/gecco/DoerrLMN17}. $\plow$ chooses the mutation rate according to a power-law distribution on $[0, 1/2]$ with exponent $\beta$. More formally, denote with $X$ the r.v. (random variable) that returns the mutation rate at a given step. The power-law operator $\plow$ uses a probability distribution $D_{n/2}^{\beta}$ s.t. $\pr{X = k} = H_{n/2}^{\beta}k^{-\beta}$, where 
$H_{\ell}^{\beta} = \sum_{j = 1}^{\ell} \frac{1}{j^{\beta}}.$
The $H_{\ell}^{\beta}$s are known in the literature as generalized harmonic numbers. Interestingly, generalized harmonic numbers can be approximated with the Riemann Zeta function as $\zeta(\beta) = \lim_{\ell \to + \infty} H_{\ell}^{\beta}$. In particular, harmonic numbers $H_{n/2}^{\beta}$ are always upper-bounded by a constant, for increasing problem size and for a fixed $\beta > 1$.

\subsection{Non-uniform Mutation Rates.}
In this paper we consider an alternative approach to the non-uniform mutation operators described above. For a given probability distribution
$P = [1, \dots, n] \rightarrow \mathbb{R}$
the proposed mutation operator samples an element $k \in [1, \dots, n]$ according to the distribution $P$, and flips \emph{exactly} $k$-many bits in an input string $x = (x_1, \dots  x_n)$, chosen uniformly at random among all possibilities. This framework depends on the distribution $P$, which we always assume fixed throughout the optimization process.

Based on the results of Doerr et al. \cite{DBLP:conf/gecco/DoerrLMN17}, we study a specialization of our non-uniform framework that uses a distribution of the form $P = D_{n}^{\beta}$. We refer to this operator as $\clmut$, and pseudocode is given in Algorithm~\ref{alg:myMutation}.
\begin{algorithm}[t]
	\caption{The mutation operator $\clmut (x)$}
	\label{alg:myMutation}
    \textbf{input:} a pseudo-Boolean array $x$\;
    \textbf{output:} a mutated pseudo-Boolean array $y$\;
    $\quad $\\
    $y \gets x$\;
    choose $k \in [1, \dots, n]$ with distribution $D_{n}^{\beta}$\;
	flip $k$-bits of $y$ chosen uniformly at random\;
    $\quad $\\
	\textbf{return} $y$\;
\end{algorithm}
This operator uses a power-law distribution on the probability of performing exactly $k$-bit flips in one iteration. That is, for $x\in \{0, 1\}^n$ and all $k\in\{1,\dots,n\}$, 
\begin{equation}
\label{eq:hamming_distance}
\pr{\hamming{x, \clmut(x)} = k} = \inverse{H_n^{\beta}} k^{-\beta} 
\end{equation}
We remark that with this operator, for any two points $x, y \in \{0, 1\}^n$, the probability $\pr{y = \clmut (x)}$ only depends on their Hamming distance $\hamming{x, y}$.

Although both operators, $\plow$ and $\clmut$, are defined in terms of a power-law distribution their behavior differs. We note that, for any choice of the constant $\beta > 1$ and all $x \in \{0, 1\}^n$, $\pr{\hamming{x, \plow (x)} = 0} > 0$, while $\pr{\hamming{x, \clmut (x)} = 0} = 0$. We discuss the advantages and disadvantages of these two operators in Sections~\ref{theoretical:al}.

\subsection{Submodular functions and matroids.}

Submodular set function functions intuitively capture the notion of diminishing returns -- i.e. the more you acquire the less your marginal gain. More formally, the following definition holds.
\begin{definition}
\label{def:subm_func}
A set function $f:2^V \rightarrow \mathbb{R}_{\geq 0}$ is submodular if it holds $f(S) +f(T) \geq f(S\cup T) + f(S\cap T)$ for all $S, T \subseteq V$.
\end{definition}

We remark that in this context $V$ is always a finite set. It is well-known that the defining axiom in Definition \ref{def:subm_func} is equivalent to the requirement
\begin{equation}
f(S\cup \{x\}) - f(S) \geq f(T\cup \{x\}) - f(T),\label{def:submod_1} 
\end{equation}
for all $S, T \subseteq V$ such that $S\subseteq T$ and $x \in T\setminus S$ (see eg. Welsh \cite{welsh2010matroid}). 

We say that a set function $f:2^V \rightarrow \mathbb{R}_{\geq 0}$ is \emph{symmetric} if it holds $f(S) = f(V\setminus S)$ for all $S\subseteq V$.

In some cases, feasible solutions are characterized as the independent sets of a matroid with base set $V$, as in the following definition.
\begin{definition}
\label{def:single_matr}
Given a set $V$, a matroid\\$\mathcal{M} =(V, \ind )$ with base set $V$ consists of a collection of subsets $\ind$ of $V$ with the following properties:
\begin{enumerate}[$\bullet$]
\item $\emptyset \in \ind$;
\item if $T \in \ind$, then $S\in \ind$ for all subsets $S\subseteq T$;
\item if $S, T \in \ind$ and $\absl{S} \leq \absl{T}$, then there exists a point $x \in T\setminus S$ s.t. $S\cup \{x\} \in \ind$;
\end{enumerate}
\end{definition}

From the axioms in Definition \ref{def:single_matr}, it follows that two maximal independent sets always have the same number of elements. This number is called the \emph{rank} of a matroid. It is possible to generalize this notion, as in the following definition.
\begin{definition}
\label{def:rank_function}
Consider a matroid $\mathcal{M} =(V, \ind )$. For any subset $S \subseteq V$, the rank function $r(C)$ returns the size of the largest independent set in $S$ - i.e. $r(S) = \argmax_{T \subseteq S} \{ \absl{T} \colon T \in \ind \} $.
\end{definition}

\subsection{Markov's Inequality.}
We introduce a basic probabilistic inequality that is useful in the run time analysis in Section \ref{sec:large_populations_unconstrained}. This simple tool is commonly referred to as Markov's Inequality. We use the following variation of it. 
\begin{lemma}[Markov]
\label{lemma:BD}
Let $X$ be a random variable, where $X \in [0, 1]$. Then it holds
\[
\pr{X \leq c} \leq \frac{1 - \expect{X}}{1 - c},
\]
for all $0 \leq c \leq \expect{X}$.
\end{lemma}
For a discussion of Lemma \ref{lemma:BD}, see eg. Mitzenmacher and Upfal \cite[Theorem~3.1]{Mitzenmacher:2005:PCR:1076315}.
%
%
\subsection{The Multiplicative Drift Theorem.}
The Multiplicative Drift theorem is a powerful tool to analyze the expected run time of randomized algorithms such as the \ea. Intuitively, for a fitness function $f:2^V \rightarrow \mathbb{R}_{\geq 0}$ we view the run time of the \ea as a Markov chain $\{X_t\}_{t \geq 0}$, where $X_t$ depends on the $f$-value reached at time-step $t$. The Multiplicative Drift theorem gives an upper-bound on the expected value of the first hitting time $T =  \inf \{t \colon X_t = 0\}$, provided that the change of the average value of the process $\{X_t\}_{t \geq 0}$ is within a multiplicative factor of the previous solution. The following theorem holds.
\begin{theorem}[Theorem 3 in Doerr et al. \cite{DBLP:journals/algorithmica/DoerrJW12}]
\label{thm:multiplicative_drift}
Let $\{X_t\}_{t\geq 0}$ be a random variable describing a Markov process over a finite state space $\mathcal{S} \subseteq \mathbb{R}$. Let $T$ be the random variable that denotes the earliest point in time $t \in \mathbb{N}_0$ such that~$X_t = 0$. Suppose that there exist $\delta > 0$, $c_{\min} > 0$, and $c_{\max} > 0$ such that
\begin{enumerate}[$\bullet$]
\item $\mathbb{E}[X_t - X_{t + 1} \mid X_t] \geq \delta X_t$;
\item $X_t \in [c_{\min}, c_{\max}]\cup \{0\}$;
\end{enumerate}
for all $t < T$. Then it holds $\mathbb{E}[T] \leq \frac{2}{\delta} \ln \left ( 1 + \frac{c_{\max}}{c_{\min}} \right )$.
\end{theorem}
%
%
\section{Artificial Landscapes}
\label{theoretical:al}
\subsection{General upper bounds for the \ea.}
\label{theoretical:clmut}
In this section we bound from above the run time of the \ea using the mutation operator $\clmut$ on any fitness function $f\colon \{0, 1\}^n \rightarrow \mathbb{R}$. It is well-known that the \ea using uniform mutation on any such fitness function has expected run time at most $n^n$. This upper-bound is tight, in the sense that there exists a function $f$ s.t. the expected run time of the \ea using uniform mutation to find the global optimum of $f$ is $\Omega (n^n)$. For a discussion on these bounds see Droste et al. \cite{Droste:2002:AEA:568273.568277}. Doerr et al. \cite{DBLP:conf/gecco/DoerrLMN17} prove that on any fitness function $f\colon \{0, 1\}^n \rightarrow \mathbb{R}$ the \ea using the mutation operator $\plow$ has run time at most $\bigo{H_{n/2}^{\beta}2^n n^{\beta}}$. Similarly, we derive a general upper bound on the run time of the \ea using mutation $\clmut$.
\begin{lemma}
\label{lemma:any-fitness}
On any fitness function $f\colon\{0, 1\}^n \rightarrow \mathbb{R}$ the \ea with mutation $\clmut$ finds the optimum solution after expected $\bigo{H_n^{\beta}e^{n/e}n^{\beta}}$ fitness evaluations, with the constant implicit in the asymptotic notation independent of $\beta$.
\end{lemma}
\begin{proof}
Without loss of generality we assume $n$ to be even. We proceed by identifying a general lower bound on the probability of reaching any point from any other point. To this end, let $x, y \in \{0, 1\}^n$ be any two points and let $k = \hamming{x, y}$ be their Hamming distance. Then the probability of reaching the point $y$ in one iteration from $x$ is
\begin{equation*}
\pr{y  = \clmut (x)} = \binom{n}{k}^{-1} \pr{\hamming{x, \clmut (x)} = k}.
\end{equation*}
From \eqref{eq:hamming_distance} we have that it holds $\pr{\hamming{x, \clmut (x)} = k}= \inverse{H_n^{\beta}} k^{-\beta} \geq \inverse{H_n^{\beta}} n^{-\beta}$ for all choices of $x \in \{0, 1\}^n$ and $k = 1, \dots, n$. Using a known lower bound of the binomial coefficient we have that 
\[
\binom{n}{k}^{-1} \geq \binom{n}{n/2}^{-1} \geq (2 e)^{-n/2} \geq e^{-n/e}, 
\]
from which it follows that $\pr{y = \clmut (x)} \geq \inverse{H_n^{\beta}}  e^{-n/e} n^{-\beta}$, for any choice of $x$ and $y$. We can roughly estimate run time as a geometric distribution with probability of success $\pr{y = \clmut (x)}$. Hence, we conclude by taking the inverse of the estimate above, which yields an upper-bound on the probability of convergence on any fitness function.
\end{proof}

We consider the \onemaxm function, defined as 
$\mbox{\onemaxm}(x_1, \dots, x_n)= \onemax{x} = \sum_{j = 1}^{n}x_j$. This simple linear function of unitation returns the number of ones in a pseudo-Boolean input string. The \ea with mutation operators \umut{} and $\plow$ finds the global optimum after $\bigo{n \log n}$ fitness evaluations (see \cite{DBLP:conf/ppsn/Muhlenbein92,Droste:2002:AEA:568273.568277,DBLP:conf/gecco/DoerrLMN17}). It can be easily shown that the \ea with mutation operator $\clmut$ achieves similar performance on this instance.
\begin{lemma}
\label{lemma:clmut_onemax}
The \ea with mutation $\clmut$ finds the global optimum of the \onemaxm after expected $\bigo{H_n^{\beta}n \log n}$ fitness evaluations, for all $\beta > 1$ and  with the constant implicit in the asymptotic notation independent of $\beta$.
\end{lemma}
\begin{proof}
We use the fitness level method outlined in Wegener \cite{Wegener:2001:TAE:646254.684097}. Define the levels $A_i = \{x \in \{0, 1\}^n \colon f(x) = i\}$, and consider the quantities $s_i = (n-i) \inverse{n H_n^{\beta}}$, for all $i = 1, \dots, n$. Then each $s_i$ is a lower-bound on the probability of reaching a higher fitness in one iteration. Denote with $T_{\clmut}(f)$ the run time of the \ea with mutation $\clmut$ on the function $f = $\onemaxm. By the fitness level theorem, we obtain an upper-bound on the run time as
\[
T_{\clmut}(f) \leq \sum_{i = 0}^{n - 1}\frac{1}{s_i} \leq H_n^{\beta} n \int_{0}^{n - 1} \frac{dx}{n - x} \leq H_n^{\beta} n \log n
\]
and the claim follows.
\end{proof}
\subsection{A comparison with static uniform mutations.}
\label{sec:jump}
Droste et al.~\cite{Droste:2002:AEA:568273.568277} defined the following jump function.
\[ 
\jump(x)=
\left \{
\begin{array}{ll}
m + \onemax{x} & \mbox{ if } \onemax{x} \leq n - m; \\
m + \onemax{x}  &\mbox{ if } \onemax{x} = n;\\
n - \onemax{x} & \mbox{ otherwise.}
\end{array}
\right .
\]
For $1 < m < n$ this function exhibits a single \emph{local} maximum and a single \emph{global} maximum. The first parameter of $\jump$ determines the Hamming distance between the local and the global optimum, while the second parameter denotes the size of the input. We present a general upper-bound on the run time of the \ea on $\jump$ with mutation operator $\clmut$. Then, following the footsteps of Doerr et al.~\cite{DBLP:conf/gecco/DoerrLMN17}, we compare the performance of $\clmut$ with static mutation operators on jump functions for all $m \leq n/2$. 

\begin{lemma}
\label{lemma:generalJump}
Consider a jump function $f = \jump$ and denote with $T_{\clmut}(f)$ the expected run time of the \ea using the mutation $\clmut$ on the function $f$. 
$T_{\clmut}(f) = H_n^{\beta} \binom{n}{m}\, \bigo{m^{\beta}}$,
were the constant implicit in the asymptotic notation is independent of $m$ and $\beta$.
\end{lemma}
\begin{proof}
We use the fitness level method. Define the levels $A_i = \left \{x \in \{0, 1\}^n \colon f(x) = i \right \}$ for all $i = 1, \dots, n$, and consider the quantities
\[
s_i = \left \{
\begin{array}{l}
(n-i) \inverse{n H_n^{\beta}},\ 0 \leq i \leq n-m-1;\\
\binom{n}{m}^{-1} \inverse{H_n^{\beta}} m^{-\beta},\ i=n-m;\\
i \inverse{n H_n^{\beta}},\ n-m+1 \leq i \leq n-1;
\end{array}
\right .
\]
Then each $s_i$ is a lower bound for the probability of reaching a higher fitness in one iteration from the level $A_i$. By the fitness level theorem we obtain an upper bound on the run time as
\begin{align*}
T_{\clmut}(f) & \leq  \binom{n}{m} H_n^{\beta} m^{\beta} + \sum_{i = 0}^{n - m - 1}\frac{n H_n^{\beta}}{n - i} + \sum_{i = n - m + 1}^{n - 1}\frac{n H_n^{\beta}}{i}\\
& \leq \binom{n}{m} H_n^{\beta} m^{\beta} + 2 n H_n^{\beta} \int_m^n \frac{dx}{x} = \binom{n}{m} H_n^{\beta} m^{\beta} + 2 n H_n^{\beta} \ln \frac{n}{m},
\end{align*}
for any choice of $\beta > 1$. Since we have that $1 < m < n$ and $m \geq 2$, then it follows that
\[
2 n H_n^{\beta} \ln \frac{n}{m} \leq 2 n H_n^{\beta} \ln n\leq 2H_n^{\beta} \binom{n}{m},
\]
and the lemma follows.
\end{proof}

Note that the upper-bound on the run time given in Lemma \ref{lemma:generalJump} yields polynomial run time on all functions $\jump$ with $m$ constant for increasing problem size and also with $n - m$ constant for increasing problem size. 

Following the analysis of Doerr et al. \cite{DBLP:conf/gecco/DoerrLMN17}, we can compare the run time of the \ea with mutation $\clmut$ with the \ea with uniform mutations, on the jump function $\jump$ for $m \leq n/2$.

\begin{corollary}
Consider a jump function $f = \jump$ with $m \leq n/2$ and denote with $T_{\clmut}(f)$ the run time of the \ea using the mutation $\clmut$ on the function $f$. Similarly, denote with $T_{\optmut}(f)$ the run time of the \ea using the best possible static uniform mutation on the function $f$. Then it holds $T_{\clmut} (f) \leq c m^{\beta - 0.5} \, H_n^{\beta} \, T_{\optmut}(f)$, for a constant $c$ independent of $m$ and $\beta$.
\end{corollary}

The result above holds because Doerr et al. \cite{DBLP:conf/gecco/DoerrLMN17} prove that the best possible optimization time for a static mutation rate a function $f = \jump $ with $m \leq n/2$ is lower-bounded as $1/2 \, n^m/m^m \, (n/(n - m))^{n - m} \leq T_{\optmut}(f)$.
%
%
\section{The Unconstrained Submodular Maximization Problem}\label{sec:submodular}
We study the problem of maximizing a non-negative submodular function $f:2^V \rightarrow \mathbb{R}_{\geq 0}$ with no side constraints. More formally, we study the problem
\begin{equation}
\label{probl_1}
\mbox{argmax}_{C \subseteq V}f(C).
\end{equation}

This problem is $\mathsf{APX}$-complete. That is, this problem is $\mathsf{NP}$-hard and  does not admit a polynomial time approximation scheme (PTAS), unless $\mathsf{P} = \mathsf{NP}$ (see Nemhauser and Wolsey \cite{DBLP:journals/mor/NemhauserW78}).

We denote with $\opt$ any solution of Problem \eqref{probl_1}, and we denote with $n$ the size of $V$.

\subsection{Heavy-tailed mutations are useful.}
\label{sec:submodular-bound}

We prove that the \ea with mutation $\clmut$ is a $(1/3 - \varepsilon /n)$-approximation algorithm for Problem \ref{probl_1}. In our analysis we assume neither monotonicity nor symmetry. We approach this problem by searching for $(1 + \alpha)$-local optima, which we define below.
\begin{definition}
\label{def:local_optima}
Let $f\colon 2^V\rightarrow \mathbb{R}_{\geq 0}$ be any submodular function. A set $\lm \subseteq V$ is a $(1 + \alpha)$-local optimum if it holds $(1 + \alpha)f(\lm) \geq f(\lm \setminus \{u\})$ for all $u \in \lm$, and $(1 + \alpha)f(\lm) \geq f(\lm \cup \{ v\})$ for all $v \in V \setminus \lm$, for a constant $\alpha > 0$.
\end{definition}

This definition is useful in the analysis because it can be proved that either $(1 + \alpha)$-local optima or their complement always yield a good approximation of the global maximum, as in the following theorem.
\begin{theorem}[Theorem 3.4 in Feige et al.~\cite{DBLP:journals/siamcomp/FeigeMV11}]
\label{thm:feige}
Consider a non-negative submodular function $f\colon 2^V\rightarrow \mathbb{R}_{\geq 0}$ and let $\lm$ be a $(1 + \varepsilon/n^2)$-local optimum as in Definition \ref{def:local_optima}. Then either $\lm$ or $V \setminus \lm$ is a $(1/3 - \varepsilon/n)$-approximation of the global maximum of $f$.  
\end{theorem}
It is possible to construct examples of submodular functions that exhibit $(1 + \varepsilon/n^2)$-local optima with arbitrarily bad approximation ratios. Thus, $(1 + \varepsilon/n^2)$-local optima alone do not yield any approximation guarantee for Problem \eqref{probl_1}, unless the fitness function is symmetric.

We can use Theorem \ref{thm:feige} to estimate the run time of the \ea using mutation $\clmut$ to maximize a given submodular function. Intuitively, it is always possible to find a $(1 + \varepsilon/n^2)$-local optimum in polynomial time using single bit-flips. It is then possible to compare the approximate local solution $S$ with its complement $V\setminus S$ by flipping all bits in one iteration.

We do not perform the analysis on a given submodular function $f$ directly, but we consider a corresponding \emph{potential function} $g_{f, \varepsilon}$ instead. We define potential functions as in the following lemma.
\begin{lemma}
\label{lemma:potential_function}
Consider a non-negative submodular function $f\colon 2^V\rightarrow \mathbb{R}_{\geq 0}$. Consider the function $g_{f,\varepsilon}(U) = f(U) + \varepsilon \frac{\opt}{n}$, for all $U \subseteq V$. The following conditions hold
\begin{enumerate}[(1)]
\item $g_{f,\varepsilon}(U)$ is submodular.
\item $g_{f,\varepsilon}(U) \geq \varepsilon \,\opt/n$, for all subsets $U \subseteq V$.
\item Suppose that a solution $U \subseteq V$ is a $\delta$-approximation for $g_{f,\varepsilon}$, for a constant $0 < \delta < 1$. Then $U$ is a $(\delta - \varepsilon/n)$-approximation for $f$.
\end{enumerate}
\end{lemma}
\begin{proof}
\emph{(1)} The submodularity of $g_{f,\varepsilon}(U)$ follows immediately from the fact that $f(U)$ is submodular, together with the fact that the term $\varepsilon \opt/n$ is constant. \emph{(2)} The property follows directly from the definition of $g_{f,\varepsilon}(U)$, together with the assumption that $f$ is non-negative. \emph{(3)} Fix a subset $U \subseteq V$ that is an $\delta$-approximation for $g_{f,\varepsilon}$. Then we have that
\[
 g_{f,\varepsilon}(U) \geq \delta \left ( \opt + \varepsilon \frac{\opt}{n} \right )\Rightarrow f(U) \geq \delta \left (\opt + \varepsilon \frac{\opt}{n} \right ) - \varepsilon \frac{\opt}{n}.
\]
It follows that
\[
f(U) \geq \delta \opt - (1 - \delta) \varepsilon \frac{\opt}{n} \geq \delta \opt - \varepsilon \frac{\opt}{n},
\]
where the last inequality follows from the assumption that $0 < \delta < 1$. The lemma follows.
\end{proof}
Using potential functions and their properties, we can prove the following result.
\begin{theorem} 
\label{thm:maxCut}
The \ea with mutation $\clmut$ is a $(1/3 - \varepsilon/n)$-approximation algorithm for Problem \eqref{probl_1}. Its expected run time is $\bigo{\frac{1}{\varepsilon} n^3 \log \frac{n}{\varepsilon} + n^{\beta}}$.
\end{theorem}
\begin{proof}
We prove that for all $\varepsilon > 0$, the \ea with mutation $\clmut$ finds a $(1/3 - \varepsilon/n)$-approximation of $g_{f, \varepsilon}$ (as in Lemma \ref{lemma:potential_function}) within expected $\bigo{\frac{1}{\varepsilon}n^3 \log \frac{n}{\varepsilon} + n^{\beta}}$ fitness evaluations. We then use this knowledge to conclude that the \ea with mutation $\clmut$ finds a $(1/3 - 2 \varepsilon/n)$-approximation of $f$ within $ \bigo{n^{\beta} + \frac{1}{\varepsilon}n^3 \log \frac{n}{\varepsilon}}$ fitness evaluations and the theorem follows.

We divide the run time in two phases. During (Phase 1), the \ea finds a $(1 + \varepsilon /n^2)$-local optimum of $g_{f, \varepsilon }$. During (Phase 2) the algorithm finds a $(1/3 - \varepsilon/n)$-approximation of the global optimum of $g_f$ using the heavy-tailed mutation.

\emph{(Phase 1)} We use the multiplicative increase method. Denote with $x_t$ the solution found by the \ea at time step $t$, for all $t \geq 0$. Then for any solution $x_t$ it is always possible to make an improvement of $(1 + \varepsilon/n^2)g_{f, \varepsilon}(x_t)$ on the fitness in the next iteration, by adding or removing a single vertex, unless $x_t$ is already a $(1 + \varepsilon /n^2)$-local optimum. We refer to any single bit-flip that yields such an improvement of a fitness as \emph{favorable bit-flip}. We give an upper-bound on the number of favorable bit-flips $k$ to reach a $(1 + \varepsilon /n^2)$-local optimum, by solving the following equation
\[
\left ( 1 + \frac{\varepsilon}{n^2} \right )^k \varepsilon \frac{\opt}{n} \leq \opt + \varepsilon \frac{\opt}{n}\Leftrightarrow \left ( 1 + \frac{\varepsilon}{n^2} \right )^k \leq \frac{n}{\varepsilon} + 1,
\]
where we have used that that for the initial solution $x_0$, $g_{f, \varepsilon} (x_0) \geq \varepsilon \opt / n$ (see Lemma \ref{lemma:potential_function}(2)). From solving this inequality it follows that the \ea with mutation $\clmut$ reaches a $(1 + \varepsilon/n^2)$-local maximum after at most $k = \bigo{\frac{1}{\varepsilon}n^2 \log \frac{n}{\varepsilon}}$ favorable bit-flips. Since the probability of performing a single chosen bit-flip is at least $\inverse{H_n^{\beta}}n^{-1} = \Omega(1/n)$, then the expected waiting time for a favorable bit-flip to occur is $\bigo{n}$, we can upper-bound the expected run time in this initial phase as $\bigo{\frac{1}{\varepsilon} n^3\log \frac{n}{\varepsilon}}$.

\emph{(Phase 2)} Assume that a $(1 + \varepsilon/n^2)$-local optimum has been found. Then from Theorem \ref{thm:feige} follows that either this local optimum or its complement is a $(1/3 - \varepsilon /n)$-approximation of the global maximum. Thus, if the solution found in Phase $1$ does not yield the desired approximation ratio, a $n$-bit flip is sufficient to find a $(1/3 - \varepsilon /n)$-approximation of the global optimum of $g_f$. The probability of this event to occur is at least $\inverse{H_n^{\beta}}n^{-\beta} = \Omega (n^{-\beta})$ by \eqref{eq:hamming_distance}. After an additional phase of expected $\bigo{n^{\beta}}$ fitness evaluations the \ea with mutation $\clmut$ reaches the desired approximation of the global maximum.
\end{proof}
%
%
\subsection{An improved upper-bound on the run time.}\label{sec:large_populations_unconstrained}
We prove that the \ea with mutation $\clmut$ yields an improved upper-bound on the run time over that of Theorem \ref{thm:maxCut}, at least with constant probability. This upper-bound yields an improvement over the run time analysis of a standard deterministic Local Search (LS) algorithm (see Theorem 3.4 Feige et al.~\cite{DBLP:journals/siamcomp/FeigeMV11}), at least with constant probability.
\begin{theorem}[Theorem 2.1 in Feige et al.~\cite{DBLP:journals/siamcomp/FeigeMV11}]
\label{thm:vondrak_new}
Let $f:2^V  \rightarrow \mathbb{R}_{\geq 0}$ be a submodular function, and denote with $R\subseteq V$ a set chosen uniformly at random. Then $\expect{f(R)}\geq \opt / 4$.
\end{theorem}
We exploit this result to obtain an improved upper-bound on the run time. Intuitively, the initial solution sampled by the \ea yields a constant-factor approximation guarantee at least with constant probability. We can use this result to prove the following theorem.
\begin{theorem} 
\label{thm:maxCut2}
The \ea with mutation $\clmut$ is a $(1/3 - \varepsilon/n)$-approximation algorithm for Problem \eqref{probl_1} after $\bigo{\frac{1}{\varepsilon} n^3 + n^{\beta}}$ fitness evaluations, at least w.c.p.
\end{theorem}
\begin{proof}
This proof is similar to that of Theorem \ref{thm:maxCut}. We denote with $x_t$ a solution reached by the \ea at time step $t$. We first prove that the definition of submodularity implies that, with high probability the initial solution $x_0$ yields a constant-factor approximation guarantee. We then perform a run time analysis as in Theorem \ref{thm:maxCut}, by counting the expected time until the fittest individual is chosen for selection, and a local improvement of at least $(1 + \epsilon /n^2)$ is made, assuming that the initial solution yields a constant-factor approximation guarantee.

Denote with $R\subseteq V$ a set chosen uniformly at random and fix a constant $\delta > 1$. We combine Theorem \ref{thm:vondrak_new} with Lemma \ref{lemma:BD}, by choosing $X= f(R)/\opt$ and obtain,
\[
\pr{f(R) \leq \frac{1}{4\delta} \opt} = \pr{X \leq \frac{1}{4\delta}}\leq \frac{1 - 1/4}{1 - 1/4\delta} = \frac{3 \delta}{4 \delta - 1},
\]
where the last inequality by applying Theorem \ref{thm:vondrak_new} and linearity of expectation to the r.v. $X = f(R)/\opt$, to obtain that $\expect{X} \geq 1/4$. We have,
\[
\pr{x_0 > \frac{1}{4 \delta} \opt}\geq 1 - \pr{f(R) \leq \frac{1}{4 \delta} \opt} \geq 1 - \frac{3 \delta}{4 \delta - 1}.
\]
In the following, for a fixed constant $\delta > 1$, we perform the run time analysis as in Theorem \ref{thm:maxCut} conditional on the event $\mathcal{A} = \{x_0 > \opt / 4 \delta \}$, which occurs at least w.c.p.

Again, we divide the run time in two phases. During Phase 1, the \ea finds a $(1 + \varepsilon /n^2)$-local optimum of $f$. During Phase 2 the algorithm finds a $(1/3 - \varepsilon/n)$-approximation of the global optimum of $f$ using the heavy-tailed mutation.

\emph{(Phase 1)} For any solution $x_t$ it is always possible to make an improvement of $(1 + \varepsilon/n^2)f(x_t)$ on the fitness in the next iteration, by adding or removing a single vertex -- the favorable bit-flip, unless $x_t$ is already a $(1 + \varepsilon /n^2)$-local optimum. Again, we give an upper-bound on the number of favorable bit-flips $k$ to reach a $(1 + \varepsilon /n^2)$-local optimum, by solving the following equation
\[
\left ( 1 + \frac{\varepsilon}{n^2} \right )^k \frac{\opt}{4 \delta} \leq \opt  \Longleftrightarrow \left ( 1 + \frac{\varepsilon}{n^2} \right )^k \leq 4 \delta,
\]
from which it follows that the \ea with mutation $\clmut$ reaches a $(1 + \varepsilon/n^2)$-local maximum after at most $k = \bigo{\frac{1}{\varepsilon}n^2}$ favorable moves. Since the probability of performing a single chosen bit-flip is at least $\inverse{H_n^{\beta}}n^{-1} = \Omega(1/n)$, then the expected waiting time for a favorable bit-flip to occur is $\bigo{n}$, we can upper-bound the expected run time in this initial phase as $\bigo{\frac{1}{\varepsilon} n^3}$.

\emph{(Phase 2)} We conclude applying the heavy-tailed mutation step: If the solution found in Phase $1$ does not yield the desired approximation ratio, a $n$-bit flip is sufficient to find a $(1/3 - \varepsilon /n)$-approximation of the global optimum of $f$. The probability of this event to occur is at least $\inverse{H_n^{\beta}}n^{-\beta} = \Omega (n^{-\beta})$ by \eqref{eq:hamming_distance}. After an additional phase of expected $\bigo{ n^{\beta}}$ fitness evaluations the \ea with mutation $\clmut$ performs an $n$-nit flip, thus reaching the desired approximation ratio.
\end{proof}
%
%
\section{Symmetric Submodular Functions under a Matroid Constraint.}\label{sec:submodular_const}

In this section we consider the problem of maximizing a non-negative submodular function $f:2^V \rightarrow \mathbb{R}_{\geq 0}$ under a single matroid constraint $\mathcal{M} = (V, \ind)$. More formally, we study the problem
\begin{equation}
\label{probl_2}
\mbox{argmax}_{C \in \ind }f(C).
\end{equation}
We denote with $\opt$ any solution of Problem \eqref{probl_2}, and we denote with $n$ the size of $V$. Note that this definition of $\opt$ differs from that of Section \ref{sec:submodular}.

We approach this problem, by maximizing the following fitness function
\begin{equation}
z_f(C) = 
\left \{
\begin{array}{ll}
f(C) & \mbox{if } C \in \ind; \\
r(C) - \absl{C} & \mbox{otherwise};
\end{array}
\right .
\label{eq:fitness_function}
\end{equation}
with $r$ the rank function as in Definition \ref{def:rank_function}. If a solution $C$ is unfeasible, then $z_f(C)$ returns a negative number, whereas if $C$ is feasible, then $z_f(C)$ outputs a non-negative number.

When studying additional constraints on the solution space the problem becomes more involved, so we require a different notion of local optimality.
\begin{definition}
\label{def:local_optima2}
Let $f\colon 2^V\rightarrow \mathbb{R}_{\geq 0}$ be a submodular function, let $\mathcal{M} = (V, \ind)$ be a matroid and let $\alpha>0$. A set $\lm \in \ind$ is a $(1 + \alpha)$-local optimum if the following hold. 
\begin{itemize}
\item[$\bullet$] $(1 + \alpha)f(\lm) \geq f(\lm \setminus \{u\})$ for all $u \in \lm$;
\item[$\bullet$] $(1 + \alpha)f(\lm) \geq f(\lm \cup \{ v\})$ for all $v \in V \setminus \lm$ s.t. $S\cup \{v\} \in \ind$;
\item[$\bullet$] $(1 + \alpha)f(\lm) \geq f((\lm  \setminus \{u\} )\cup \{ v\})$ for all $u \in \lm$ and $v \in V \setminus \lm$ s.t. $(S\setminus \{u\})\cup \{v\} \in \ind$.
\end{itemize}
\end{definition}
We prove that, in the case of a \emph{symmetric} submodular function, a $(1 + \alpha)$-local optimum as in Definition \ref{def:local_optima2} yields a constant-factor approximation ratio. To this end, We make use of the following well-known result.
\begin{theorem}[Theorem 1 in Lee et al. \cite{Lee:2009:NSM:1536414.1536459}]
\label{thm:lee_et_al}
Let $\mathcal{M} =(V, \ind)$ be a matroid and $I, J \in \ind$ be two independent sets. Then there is a mapping $\pi : J\setminus I \rightarrow (I\setminus J)\cup \{\emptyset \}$ such that
\begin{enumerate}
\item[$\bullet$] $(I\setminus \{\pi(b)\}) \cup \{b\} \in \ind$ for all $b\in J\setminus I$;
\item[$\bullet$] $\absl{\pi^{-1}(e)} \leq 1$ for all $e \in I \setminus J$.
\end{enumerate}
\end{theorem}
A mild revision of Lemma 1 and Theorem 3 in Lee et al. \cite{Lee:2009:NSM:1536414.1536459}, the following lemma holds.
\begin{lemma}
\label{lemma:matroid_local_optimality}
Consider a non-negative symmetric submodular function $f:2^V \rightarrow \mathbb{R}_{\geq 0}$, a matroid $\mathcal{M} =(V, \ind)$ and let $S$ be a $(1 + \varepsilon/n^2)$-local optimum as in Definition \ref{def:local_optima2}. Then $S$ is a $(1/3 - \epsilon/n)$-approximation for Problem \eqref{probl_2}.
\end{lemma}
\begin{proof}
Fix a constant $\varepsilon > 0$ and a set $C\in \ind$. Consider a mapping $\pi : C \setminus S \rightarrow (S \setminus C)\cup\{\emptyset\}$ as in Theorem \ref{thm:lee_et_al}. Since $S$ is a $(1 + \varepsilon/n^2)$-local optimum it holds
\begin{equation}
\left (1 + \frac{\varepsilon}{n^2} \right ) f(S) \geq f((S\setminus \{\pi(b)\}) \cup b)\label{eq:real_lemma2};
\end{equation}
for all $b \in C\setminus S$. Thus, it holds
\begin{align}
f((S & \cup \{b\}) - f(S) \nonumber \\
& \leq f((S\setminus \{\pi(b)\})\cup \{b\}) - f(S\setminus \{\pi(b)\})\nonumber \\
& \leq \left (1 + \frac{\varepsilon}{n^2} \right ) f(S) - f(S\setminus \{\pi(b)\}),\nonumber
\end{align}
where the first inequality follows from \eqref{def:submod_1}, and the second one follows from \eqref{eq:real_lemma2}. Summing these inequalities for each $b \in C\setminus S$ and using submodularity as in \eqref{def:submod_1} we obtain,
\begin{align}
f(S & \cup C) - f(S) \nonumber\\
& \leq  \sum_{b \in C\setminus S}\left [f(S\cup \{b\}) - f(S) \right ] \nonumber\\
& \leq \sum_{b \in C\setminus S} \left [ \left (1 + \frac{\varepsilon}{n^2} \right ) f(S) - f(S\setminus \{\pi(b)\}) \right ]. \nonumber
\end{align}
Consider a given order of the elements in $b\in C\setminus S$, i.e. $C\setminus S = \{b_1, \dots, b_k\}$. Then it holds
\begin{align*}
\sum_{b \in C\setminus S} & \left [ \left (1 + \frac{\varepsilon}{n^2} \right ) f(S)  - f(S\setminus \{\pi(b)\})  \right ] = \sum_{j = 1}^k [f(S) - f(S\setminus \{\pi(b_j)\})] + k\frac{\varepsilon}{n^2} f(S) \\
& \leq \sum_{j = 2}^k f\left ( (S \cap C) \bigcup_{\ell = 1}^j \{\pi (b_\ell)\}  \right ) - \sum_{j = 2}^k  f\left ( (S \cap C) \bigcup_{\ell = 1}^{j-1} \{\pi (b_\ell)\}  \right )\\
& + f((S\cap C) \cup \{\pi (b_1)\}) - f(S\cap C) + k \frac{\varepsilon}{n^2}f(S) \leq \left (1 + \frac{\varepsilon}{n} \right )f(S) - f(S\cap C)
\end{align*}
where the first inequality follows from \eqref{def:submod_1} and the second inequality follows by taking the telescopic sum together with the fact that $k \leq n$. Thus, it follows that
\[
2 \left (1 + \frac{\varepsilon}{n} \right ) f(S) \geq f(S\cup C) + f(S \cap C),
\]
Since $f$ is symmetric, $f(S) = f(V \setminus S)$ and we have that,
\begin{align}
3 \left (1 + \frac{\varepsilon}{n} \right ) f(S) & \geq f(\overline{S}) + f(S\cup C) + f(S \cap C)\nonumber \\
& \geq f(C\setminus S) + f(C \cap S) \geq f(C).\nonumber 
\end{align}
The claim follows by choosing $C = \opt$.
\end{proof}
We use Lemma \ref{lemma:matroid_local_optimality} to perform a run time analysis of the \ea. We consider the case of the $\clmut$ mutation, although our proof easily extends to the standard uniform mutation and $\plow$. We experimentally compare these operators in Section \ref{sec:symmetric_mutual_information}. We perform the analysis by estimating the expected run time until a $(1 + \varepsilon/n^2)$-local optimum is reached, and apply Lemma \ref{lemma:matroid_local_optimality} to obtain the desired approximation guarantee. Our analysis yields an improved upper-bound on the run time over that of Friedrich and Neumann \cite{DBLP:journals/ec/FriedrichN15}. The following theorem holds.
\begin{theorem}
\label{thm:general_bound_constrained}
The \ea with mutation $\clmut$ is a $(1/3 - \varepsilon/n)$-approximation algorithm for Problem \eqref{probl_2}. Its expected run time is $\bigo{\frac{1}{\epsilon}n^4 \log \frac{n}{\varepsilon}}$.
\end{theorem}
\begin{proof}

We perform the analysis assuming that a fitness function as in \eqref{eq:fitness_function} is used. We divide the run time in two phases. During (Phase 1) the \ea finds a feasible solution, whereas in (Phase 2) it finds a $(1 + \varepsilon /n^2)$-local optimum, given that an independent set has been found.

\emph{(Phase 1)} Assuming that the initial solution is not an independent set, then the \ea maximizes the function $r(C) - \absl{C}$, until a feasible solution is found. This is equivalent to minimizing the function $\absl{C} - r(C)$. We estimate the run time using the multiplicative drift theorem (Theorem \ref{thm:multiplicative_drift}). Denote with $x_t$ a solution found by the \ea after $t$ steps, consider the Markov chain $X_t = \absl{x_t} - r(x_t)$ and consider the first hitting time $T= \min \{ t \colon X_t = 0\}$. Then it holds $X_t \in \{0\}\cup [1, n]$. Moreover, since the probability of removing a single chosen bit-flip from the current solution is $1/en$, we have,
$
\expect{X_t - X_{t + 1} \mid X_t} \geq \frac{X_t}{en}.
$
Theorem \ref{thm:multiplicative_drift} now yields, 
$
\expect{T} \leq 2 e n \log (1 + n) .
$
We conclude that we can upper-bound the run time in this initial phase as $\bigo{n \log n}$.

\emph{(Phase 2)} We estimate the run time in this phase with the multiplicative increase method. Assuming that a feasible solution is reached, then all subsequent solutions are feasible, since $z_f(C) \geq 0$ for all feasible solutions and $z_f(C) < 0$ for all infeasible solutions.

To estimate the run time in this phase we do not perform the analysis on $f$ directly but we consider the potential function $g_{f, \varepsilon}$ from Lemma~\ref{lemma:potential_function} (recall that in this case \opt is not the global optimum of $f$, but the highest $f$-value among all feasible solutions). We prove that for all $\varepsilon > 0$, the \ea with mutation $\clmut$ finds a $(1/3 - \varepsilon/n)$-approximation of 
$
g_{f, \varepsilon}(S) = f(S) + \frac{\opt}{\varepsilon}
$,
within expected $\bigo{\frac{1}{\varepsilon}n^4 \log \frac{n}{\varepsilon}}$ fitness evaluations. We apply Lemma \ref{lemma:potential_function}(3) and conclude that the \ea with mutation $\clmut$ finds a $(1/3 - 2 \varepsilon/n)$-approximation of $f$ within $ \bigo{\frac{1}{\varepsilon}n^3 \log \frac{n}{\varepsilon}}$ fitness evaluations. 

Denote with $y_t$ the solution found by the \ea at time step $t + \ell$, for all $t \geq 0$, with $\ell$ the number of steps in Phase 1. In other words, $y_0$ is the first feasible solution found by the \ea, and $y_t$ is the solution found after additional $t$ steps. Again, the solutions $y_t$ are independent sets for all $t \geq 0$. For any solution $y_t$ it is always possible to make an improvement of $(1 + \varepsilon/n^2)g_{f, \varepsilon}(y_t)$ on the fitness in the next iteration, by adding or removing a single vertex, or by swapping two bits, unless $y_t$ is already a $(1 + \varepsilon /n)$-local optimum. Again, we refer to any single bit-flip or swap that yields such an improvement of a fitness as favorable move. We give an upper-bound on the number of favorable moves $k$ to reach a $(1 + \varepsilon /n)$-local optimum, by solving the following equation
\[
\left ( 1 + \frac{\varepsilon}{n^2} \right )^k \varepsilon \frac{\opt}{n} \leq \opt + \varepsilon \frac{\opt}{n} \Leftrightarrow \left ( 1 + \frac{\varepsilon}{n^2} \right )^k \leq \frac{n}{\varepsilon} + 1,
\]
where we have used that that for the initial solution $y_0$, $g_{f, \varepsilon} (y_0) \geq \varepsilon \opt / n$ (see Lemma \ref{lemma:potential_function}(2)). From solving the inequality it follows that the \ea reaches a $(1 + \varepsilon/n)$-local maximum after at most $k = \bigo{\frac{1}{\varepsilon}n^2 \log \frac{n}{\varepsilon}}$ favorable moves. Since the probability of performing a single chosen bit-flip or a swap is at least $H_n^{-\beta} 2^{-\beta}n^{-2}$, then the expected waiting time for a favorable bit-flip to occur is at most $ \bigo{n^2}$, hence we can upper-bound the expected run time in Phase 2 as $\bigo{\frac{1}{\varepsilon} n^4\log \frac{n}{\varepsilon}}$. 
\end{proof}

\section{Experiments}
\subsection{The Maximum Directed Cut problem.}
\label{sec:max_di_cut}

Given a directed graph $G =(V, E)$, we consider the problem of finding a subset $U \subseteq V$ of nodes such that the sum of the outer edges of $U$ is maximal. This problem is the maximum directed cut problem (\maxdicut) and is a known to be $\mathsf{NP}$-complete. 

For each subset of nodes $U \subseteq V$, consider the set $\Delta(U) = \{ (e_1, e_2) \in E \colon e_1 \in U \mbox{ and } e_2 \notin U  \}$ of all edges leaving $U$. We define the cut function $f\colon 2^V \longrightarrow \mathbb{R}_{\geq 0}$ as 
\begin{equation}
\label{eq:maxdicut_fitness}
f(U) = \absl{\Delta(U)}.
\end{equation}
The \maxdicut can be approached by maximizing the cut function as in \eqref{eq:maxdicut_fitness}. Note that this function is non-negative. Moreover, it is always submodular and, in general, non-monotone (see e.g. Feige et al. \cite{DBLP:journals/siamcomp/FeigeMV11} and Friedrich et al. \cite{DBLP:journals/ec/FriedrichN15}). Hence, this approach to the \maxdicut can be formalized as in Problem \eqref{probl_1} in Section \ref{sec:submodular}. 

We select the 123 large instances used by Wagner~et~al.~\cite{Wagner2017fastvctuning}; the number of vertices ranges from about 379 to over 6.6 million, and the number of edges ranges from 914 to over 56 million. All 123 instances are available online~\cite{urlVC}. 

The instances come from a wide range of origins. For example, there are 14 collaboration networks (ca-*, from various sources such as Citeseer, DBLP, and also Hollywood productions), five infrastructure networks (inf-*), six interaction networks (ia-*, e.g. about email exchange), 21 general social networks (soc-*, e.g., Flickr, LastFM, Twitter, Youtube), 44 subnets of Facebook (socfb-*, mostly from different American universities), and 14 web graphs (web-*, showing the state of various subsets of the Internet at particular points in time). We take these graphs and run Algorithm~\ref{alg:ea} with seven mutation operators: $\plow$ and $\mymut$ with $\beta\in \{1.5, 2.5, 3.5\}$ and \stdmut{}.\footnote{In contrast to our earlier work~\cite{DBLP:conf/ppsn/00010Q018}, we are comparing against \stdmut{}, which performs at least one flip, thus making it a fairer comparison.} We use an intuitive bit-string representation based on vertices, and we initialize uniformly at random. Each edge has a weight of 1.

For each instance-mutation pair, we perform 100 independent runs (100\,000 evaluations each) and with an overall computation budget of 72 hours per pair. Out of the initial 123 instances 67 finish their 100 repetitions per instance within this time limit.\footnote{Source categories of the 67 instances: 2x bio-*, 6x ca-*, 5x ia-*, 2x inf-*, 1x soc-*, 40x socfb-*, 4x tech-*, 7x web-*. The largest graph is socfb-Texas84 with 36\,364 vertices and 1\,590\,651 edges.} We report on these 67 in the following. We use the average cut size achieved in the 100 runs as the basis for our analyses.

Firstly, we rank the seven approaches based on the average cut size achieved (best rank is 1, worst rank is 7). Table~\ref{tab:maxcutranks} shows the average rank achieved by the different mutation approaches. 
$\stdmut{}$ performs best at the lower budget and worst at the higher budget, which we take as a  strong indication that few bit-flips are initially helpful to quickly improve the cut size, while more flips are helpful later in the search to escape local optima. 
At the higher budget, both \fmut{\beta} and \pmut{\beta} perform better than \stdmut, independent of the parameter chosen. 
In particular, $\mymut$ clearly performs better than \fmut{\beta} at both budgets, however, while $\mymut$ with $\beta=1.5$ performs best at $10\,000$ iterations, $\mymut$ with $\beta=3.5$ performs best when the budget is $100\,000$ iterations.

\begin{table}[!tbh]
\centering
\caption{Average ranks (based on mean cut size) of at $t=10\,000$ and $t=100\,000$ iterations (lower ranks are better).
}\label{tab:maxcutranks}
\begin{tabular}{l|cc}
& \multicolumn{2}{c}{average rank} \\
mutation & $t=10,000$ & $t=100,000$ \\ \hline
\fmut{1.5}		&	4.1&	5.9 \\
\fmut{2.5}		&	5.7&	4.6	\\
\fmut{3.5}		&	6.6&	4.0	\\
\pmut{1.5}		&	2.4&	3.0	\\
\pmut{2.5}		&	3.0&	1.8	\\
\pmut{3.5}		&	4.0&	1.1	\\
\stdmut{}	&	2.1 &	6.7\\
\end{tabular}
\end{table}

To investigate the relative performance difference and the statistical significance thereof, we perform a Nemenyi two-tailed test (see Figure~\ref{fig:cd}). This test performs all-pairs comparisons on Friedman-type ranked data. The results are as expected and consistent with the average ranks reported in Table~\ref{tab:maxcutranks}. 

\begin{figure*}[!t]\centering
\begin{minipage}[b]{.49\linewidth}%
\centering\includegraphics[width=\linewidth,trim={70 0 70 0},clip]{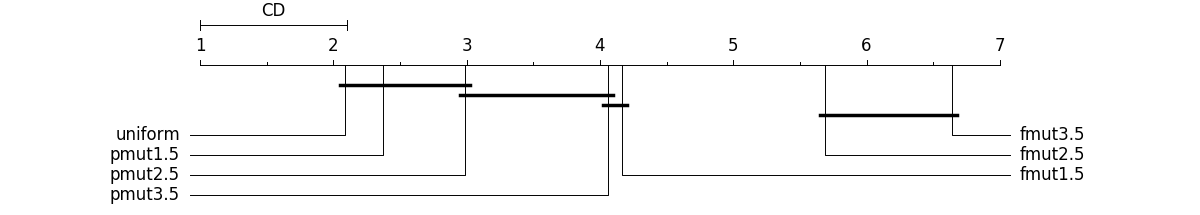}%
\subcaption{10 000 evaluations}\label{fig:cd10k}%
\end{minipage}\hspace*{4mm}
\begin{minipage}[b]{.49\linewidth}%
\centering\includegraphics[width=\linewidth,trim={70 0 70 0},clip]{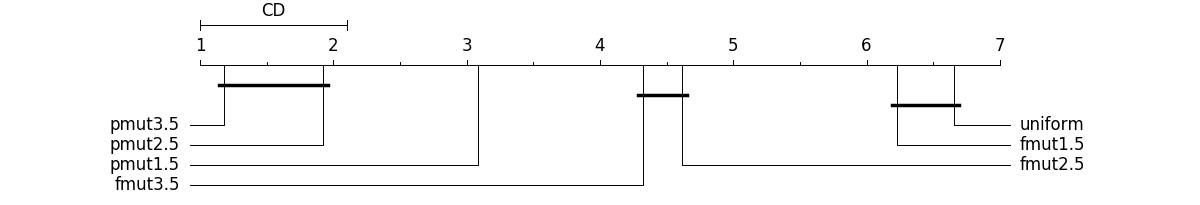}%
\subcaption{100 000 evaluations}\label{fig:cd100k}
\end{minipage}%
\caption{Critical Distance (CD) diagram based on a Nemenyi two-tailed test using the average rankings. CD (top left) shows the critical distance. Distances larger than CD corresponds to a statistical significant difference in the ranking. Relationships within a critical distance are marked with a horizontal bar.}%
\label{fig:cd}
\end{figure*}

Across the 67 instances, the achieved cut sizes vary significantly (see 
Table~\ref{tab:maxcutpercentage}). For example, the average gap between the worst and the best approach is 42.1\% at 10\,000 iterations and it still is 7.4\% at 100\,000 iterations. Also, when we compare the best $\plow$ and $\mymut$ configurations (as per Table~\ref{tab:maxcutpercentage}), then we can see that (i) $\mymut$ is better or equal to $\plow$, and (ii) the performance advantage of $\mymut$ over $\plow$ is 2.3\% and 0.8\% on average, with a maximum of 4.7\% and 6.3\% (i.e., for $10\,000$ and $100\,000$ evaluations).

\begin{table*}[t]
\centering
\caption{Summary of cut-size differences. ``total'' refers to the gap between the best and worst performing mutation out of all seven. The two highlighted pairs compare the best $\plow$ and $\mymut$ values listed in Table~\ref{tab:maxcutranks}.}
\begin{tabular}{l|cc|cc}
& \multicolumn{2}{c}{$t=10k$} & \multicolumn{2}{|c}{$t=100k$} \\
 & total & \pmut{1.5} vs \fmut{1.5} & total & \pmut{3.5} vs \fmut{3.5} \\ \hline
min gap & 0.8\% & 1.1\% & 0.0\% & 0.0\% \\
mean gap & \hspace{3mm}13.0\%\hspace{3mm} & 2.3\% & \hspace{3mm}1.9\%\hspace{3mm} &	0.8\%\\
max gap & 42.1\% &	4.7\% & 7.4\% &	6.3\%\\
\end{tabular}
\label{tab:maxcutpercentage}
\end{table*}

To investigate the extent to which mutation performance and instance features are correlated, we perform a 2D projection using a principle component analysis of the instance feature space based on the features collected from~\cite{urlVC}. We then consider the performance of the seven mutation operators at a budget of 100,000 evaluations, and we visualize it in the 2D space (see Figure~\ref{fig:footprint}). In these projections, the very dense cluster in the top left is formed exclusively by the socfb-* instances, and the ridge from the very top left to the bottom left is made up of (from top to bottom) ia-*, tech-*, web*, and ca-* instances. The ``outlier'' on the right is web-BerkStan, due to its extremely high values of the average vertex degree, the number of triangles formed by three edges ($3$-cliques), the maximum triangles formed by an edge, and the maximum $i$-core number, where an $i$-core of a graph is a maximal induced subgraph and each vertex has degree at least $i$.

Interestingly, the performance seems to be correlated with the instance features and thus, indirectly, with their origin. For example, we can see in Figure~\ref{fig:foot:uniform} that \stdmut{} does not reach a cut size that is within 1\% of the best observed average for many of the socfb-* instances (shown as many black dots in the tight socfb*-cluster). In contrast to this, \pmut{3.5}'s corresponding Figure~\ref{fig:foot:pmut35} shows only red dots, indicating that it always performs within 1\% of the best-observed. 

Lastly, we summarize the results in Figure~\ref{fig:foot:difficulty} based on the concept of instance difficulty. Here, the color denotes the number of instances that achieve a cut size within 1\% of the best observed average. Interestingly, many ia-*, ca-*, web-* and tech-* instances are solved well by many mutation operators. In contrast to this, many socfb-* instances are blue, meaning that are solved well by just very few mutation operators -- in particular, by our \pmut{3.5}.

\begin{figure*}[!t]\centering
\begin{minipage}[b]{.33\linewidth}%
\centering\includegraphics[height=33mm,trim={0 0 0 0},clip]{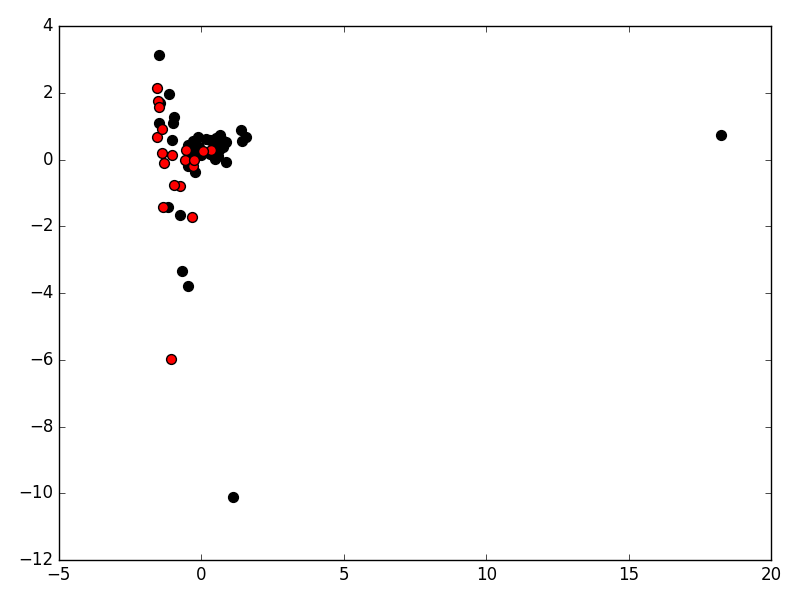}%
\subcaption{\stdmut{} footprint}\label{fig:foot:uniform}%
\end{minipage}\hspace*{0mm}
\begin{minipage}[b]{.33\linewidth}%
\centering\includegraphics[height=33mm,trim={0 0 0 0},clip]{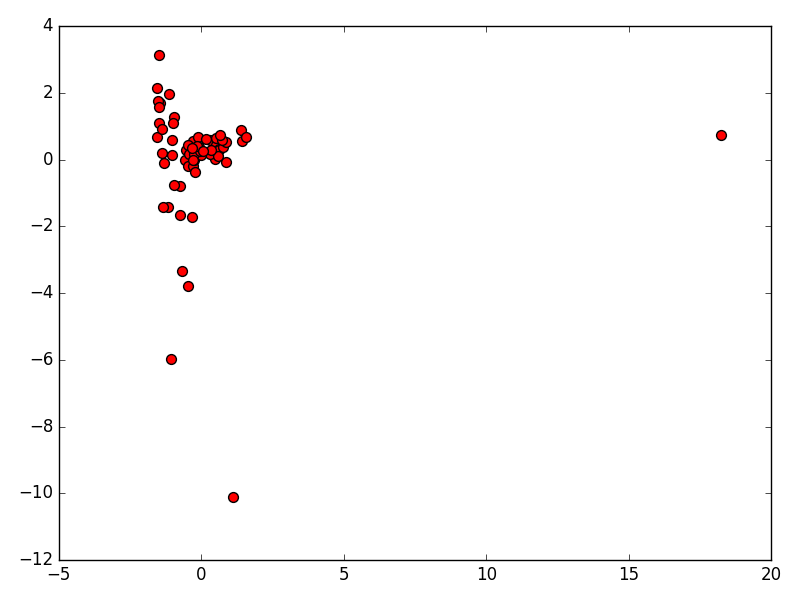}%
\subcaption{\pmut{3.5} footprint}\label{fig:foot:pmut35}%
\end{minipage}\hspace*{0mm}
\begin{minipage}[b]{.33\linewidth}%
\centering\includegraphics[height=33mm,trim={0 0 0 0},clip]{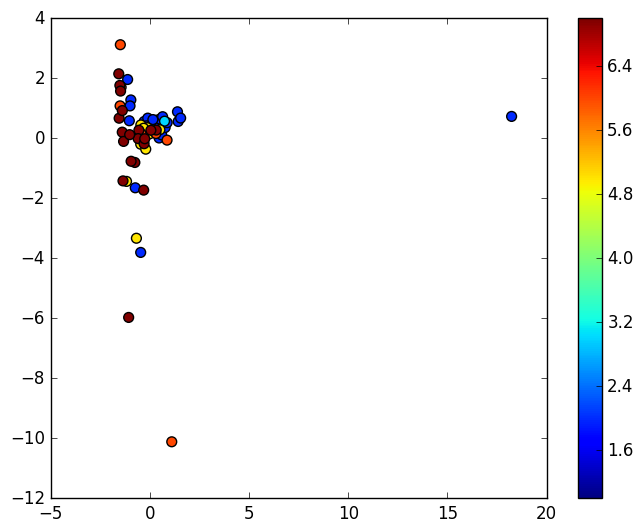}%
\subcaption{Instance difficulty}\label{fig:foot:difficulty}%
\end{minipage}%
\caption{Mutation operator footprints (left and middle plots): instances are marked red if the mutation are at most 1\% away from the best-observed performance. 
Instance difficulty (right-most plot): the color encodes the number of algorithms that perform within 1\% of the oracle performance. Note: a principle component analysis is used for the projection of the instances from the feature space into 2D.}%
\label{fig:footprint}
\end{figure*}

\subsection{The Symmetric Mutual Information problem.}
\label{sec:symmetric_mutual_information}

We study an instance of the general \emph{feature selection} problem: Given a set of observations, find a subset of relevant features (variables, predictors) for use in model construction. 

We consider the following framework. Suppose that $n$ time series $X^{(1)}, \dots, X^{(n)}$ are given, each one representing a sequence of temporal observations. For each sequence $X^{(i)}$, define the corresponding temporal variation as a sequence $Y^{(i)}$ with
$
Y^{(i)}_j = X^{(i)}_j - X^{(i)}_{j - 1}.
$

We perform feature selection of the variables $Y^{(i)}$, assuming that the joint probability distribution $p(Y^{(1)},\dots, Y^{(n)})$ is Gaussian. Specifically, given a cardinality constraint $k$, we search for a subset $S\in [n]$ of size at most $k$ s.t. the corresponding series $\chi_S := \{Y^{(i)}\colon i \in S\}$ are optimal predictors for the overall variation in the model. Variations of this setting are found in many applications (see eg. Singh et al. \cite{DBLP:journals/jair/SinghKGK09}, Zhu and Stein \cite{Zhu2006}, and Zimmerman \cite{Zimmerman})

We use the \emph{mutual information} as an optimization criterion for identifying highly informative random variables among the $\{ Y^{(i)} \}$ (see Calseton and Zidek \cite{CASELTON1984223}). For a subset $S\in [n]$, we define the corresponding mutual information as
\begin{equation}
\label{eq:mutual_information}
\mutualinfo{S} = - \frac{1}{2} \sum_{i}(1 - \rho_i^2),
\end{equation}
where the $\rho_i$ are the canonical correlations between $\chi_S$ and $\chi_{V\setminus S}$. It is well-known that the mutual information as in \eqref{eq:mutual_information} is a symmetric non-negative submodular function (see Krause et al. \cite{DBLP:journals/jmlr/KrauseSG08}). Note also that a cardinality constraint $k$ is equivalent to a matroid constraint, with independent sets all subsets $S \in [n]$ of cardinality at most $k$. Hence, this approach to feature selection consists of maximizing a non-negative symmetric submodular function under a matroid constraint, as in Problem~\eqref{probl_2}. Following the framework outlined in Section~\ref{sec:submodular_const}, we approach this problem by maximizing the following fitness function
\begin{equation}
\label{fitness:Mutual_info}
z_{\mathsf{MI}}(S) =
\left \{
\begin{array}{ll}
\mutualinfo{S} & \mbox{if } \absl{S} \leq k;\\
k - \absl{S}   & \mbox{otherwise};
\end{array}
\right .
\end{equation}

We apply this methodology to perform feature selection on an air pollution dataset (see Rhode and Muller \cite{rhode_muller}).\footnote{This dataset is publicly available at www.berkleyearth.org.} This dataset consists of hourly air \notwo~data from over 1500 sites, during a four month interval from April 5, 2014 to August 5, 2014.

For a fixed cardinality constraint $k = 200, \dots, 850$, we let the \ea with various mutation rates run for a fixed time budget at $1$K, $2.5$K, and $5$K fitness evaluations. For each set of parameters, we perform $100$ runs and take the sample mean over all resulting fitness values. We consider the \ea with uniform mutation, $\clmut$ and $\plow$ with $\beta = 1.5, 2.5, 3.5$. The results are displayed in Figure \ref{fig:result_MI}.

We observe that for a small time budget and small $k$, heavy tailed-mutations outperform the standard uniform mutation and the $\plow$. We observe that for large $k$ all mutation operators achieve similar performance. These results suggest that for small time budget, and small $k$, larger jumps are beneficial, whereas standard mutation operators may be sufficient to achieve a good approximation of the optimum, given more resources.

\begin{figure*}[!t]\centering
\centering\includegraphics[width=\linewidth]{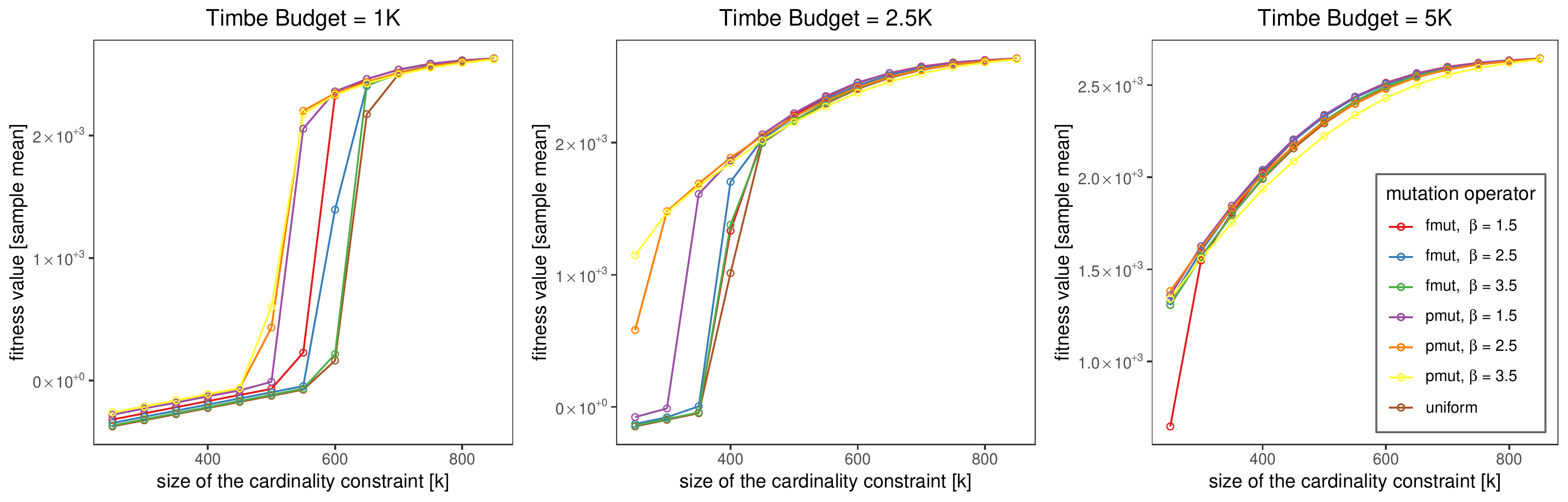}%
\caption{Solution quality achieved by the \ea with various mutation rates on a fitness function as in \eqref{fitness:Mutual_info}, for fixed cardinality constraint $k$, and varying time budget. We consider the \ea with uniform mutation, $\clmut$ and $\plow$ with $\beta = 1.5, 2.5, 3.5$. Each dot corresponds to the sample mean of $100$ independent runs.
}%
\label{fig:result_MI}
\end{figure*}

\section{Conclusions}
In the pursuit of optimizers for complex landscapes that arise in industrial problems, we have identified a new mutation operator. This operator allows for good performance of the classical \ea when optimizing not only simple artificial test functions, but the whole class of non-negative submodular functions and symmetric submodular functions under a matroid constraint. As submodular functions find applications in a variety of natural settings, it is interesting to consider the potential utility of heavy tailed operators as  building blocks for optimizers of more complex landscapes, where submodularity can be identified in parts of these landscapes.

\section*{Acknowledgment}

Markus Wagner has been supported by ARC Discovery Early Career Researcher Award DE160100850. 




\bibliography{bibliography}

%



\end{document}